\documentclass{eptcs}
\providecommand{\publicationstatus}{}
\usepackage{amsmath}
\usepackage{amsfonts}
\usepackage{makeidx} 
\usepackage{tikz}
\usepackage{cite}
\usepackage{hyperref}
\usetikzlibrary{calc}
\usetikzlibrary{shapes,automata,patterns,decorations, decorations.pathmorphing}
\usetikzlibrary{arrows}
\usetikzlibrary{intersections}
\tikzstyle{every edge}=[draw,semithick]
\tikzset{initial text={}}

\newcommand{\newterm}[1]{\emph{#1}\index{#1}}

\newcommand{\NN}{\mathbb{N}}

\newcommand{\BB}{\mathbb{B}}

\newcommand{\FALSE}{\mathsf{false}}
\newcommand{\TRUE}{\mathsf{true}}
\DeclareMathOperator{\scope}{\mathbin{.}}

\newcommand{\LTLG}{\mathsf{G}}
\newcommand{\LTLF}{\mathsf{F}}

\newcommand{\Ass}{\mathcal{A}}
\newcommand{\Gua}{\mathcal{G}}

\newcommand{\Sys}{\langle T, \tau \rangle}

\newcommand{\Sysset}{\mathcal{S}}
\newcommand{\reach}{\left<E\right>}
\newcommand{\reachinf}{\LTLG\!\reach} 
\newcommand{\inputs}{I}
\newcommand{\inalph}{\mathcal{I}}
\newcommand{\outputs}{O}
\newcommand{\outalph}{\mathcal{O}}
\newcommand{\alp}{\Sigma}
\newcommand{\lang}{L}

\newcommand{\N}{\mathbb{N}}
\newcommand{\B}{\mathbb{B}}

\usepackage{amsthm}
\newtheorem{lemma}{Lemma}
\newtheorem{definition}{Definition}
\newtheorem{theorem}{Theorem}
\usepackage{makeidx}  

\begin{document}

\newcommand{\extendedText}[2]{{#1}} 

\newcommand{\changedOne}[1]{{#1}}
\newcommand{\changedTwo}[1]{{#1}}

%
%
%
%
%
\title{Cooperative Reactive Synthesis%
%
\thanks{This work was supported in part by the Austrian Science Fund
  (FWF) through the research network RiSE (S11406-N23), by the
  European Commission through the projects STANCE (317753) and
  IMMORTAL (644905), and by the Institutional Strategy of the University of 
  Bremen, funded by the German Excellence Initiative.
  }
}
\def\titlerunning{Cooperative Reactive Synthesis}  
%
\author{Roderick Bloem
        \institute{IAIK, Graz University of\\ Technology, Graz, Austria}
\and
        R\"udiger Ehlers
        \institute{University of Bremen and\\ DFKI GmbH, Bremen, Germany}
\and
        Robert K\"onighofer
        \institute{IAIK, Graz University of\\ Technology, Graz, Austria}
}
\def\authorrunning{Roderick Bloem et al.} 
%
%

\maketitle              

\begin{abstract}
A modern approach to engineering correct-by-construction systems is to synthesize them automatically from formal specifications. 
Oftentimes, a system
can only satisfy its guarantees if certain environment assumptions hold, which motivates their inclusion in the system specification.
Experience with
modern synthesis approaches shows that 
synthesized systems tend to
satisfy their specifications by actively working towards the violation of the assumptions rather than satisfying 
assumptions and guarantees together.  Such uncooperative behavior is undesirable 
because it violates the aim of synthesis: the system should try to satisfy its guarantees and use the assumptions only when needed. Also, the assumptions often describe the valid behavior of other components in 
a bigger system, which should not be obstructed unnecessarily.

In this paper, we present a hierarchy of \emph{cooperation levels} between 
system and environment.  Each level describes how well the system enforces both 
the assumptions and guarantees.  We show how to synthesize systems that achieve 
the highest possible cooperation level for a given specification in Linear 
Temporal Logic (LTL).  The synthesized systems can also exploit cooperative 
environment behavior during operation to reach a higher cooperation level  
that is not enforceable by the system initially.  The worst-case time 
complexity of our synthesis procedure is doubly-exponential, which matches the 
complexity of standard LTL synthesis. 
\extendedText{
This is an extended version of~\cite{atva15} that features an additional appendix.
}{}

\end{abstract}

\section{Introduction}

When synthesizing reactive systems from their formal specifications, we 
typically start with a set of guarantees that the system should fulfill and
a set of assumptions about the environment. A synthesis tool then 
computes an implementation that satisfies the guarantees in all environments 
that satisfy the assumptions.
In many specifications, the system can 
influence whether the assumptions are satisfied; in particular, the system can actively force the environment to violate the 
assumptions.  The resulting implementation is correct: it fulfills its 
guarantees if the assumptions are fulfilled. However, it does so in an 
undesirable way.

Take for example a flight stabilization system that adds small forces to choices 
of forces issued by the pilot. The system guarantees state that there is little 
jitter in the absolute forces applied. This can only work for ``well-behaved'' 
evolutions of the manually selected force values, which in turn depend on 
the forces currently applied. Without the well-behavedness assumption, the system has no way 
to stabilize the flight. However, if the system has the 
opportunity to add forces that make the pilot violate this assumption, it can 
do so without violating its overall specification (namely that the guarantees 
must hold whenever the assumptions hold). Clearly, such 
behavior is not covered by the specifier's intent.

This observation leads to the question of how we can synthesize systems that 
cooperate with the environment whenever possible. Our main idea is that, at any 
point in time, the system should try to be as cooperative as possible while 
still ensuring correctness. Depending on the concrete situation, several levels 
of cooperation may be possible. Some examples are: 
\begin{enumerate}
\item The system can enforce both the assumptions and guarantees to hold.
\item The system can enforce the guarantees if the assumptions hold and 
the system can give the environment the opportunity to satisfy the assumptions at the same time. 
\item The system can neither enforce the assumptions nor the guarantees, 
but environment and system together can satisfy both. 
\end{enumerate}
The first of these levels is most beneficial: there is no need to rely on 
the environment, not even for satisfying the assumptions.  On the second level 
we can satisfy the correctness objective even without allowing the system to 
enforce an assumption violation.  From states of the third level, the system 
cannot enforce correctness.  However, instead of resigning and behaving 
arbitrarily, the system still offers some executions along which both the 
assumptions and guarantees are fulfilled, thereby optimistically assuming that 
the environment is helpful. 

In this paper, we perform a rigorous analysis of cooperation between the system 
and its environment in the setting of reactive synthesis.
We generalize the three levels of cooperation from above 
to a fine-grained \newterm{cooperation hierarchy} that allows us to reason 
about how cooperative a controller for a given specification can be. A level in 
our hierarchy describes what objectives (in terms of assumptions and 
guarantees) the controller can achieve on its own and for what objectives it has 
to rely on the environment.  

As a second contribution, we present a synthesis procedure to construct a 
controller that always picks the highest possible cooperation level in our 
hierarchy for a given linear-time temporal logic (LTL) specification.  The 
synthesized controllers do not only enforce the highest possible 
level statically, but also exploit environment behavior that enables reaching a higher level that cannot be enforced initially.  Thus, implementations 
synthesized with our approach are \newterm{maximally cooperative}, without any need to 
declare cooperation in the specification explicitly by enumerating scenarios in which a system can operate in a cooperative way. 
Our \newterm{maximally cooperative synthesis procedure} takes at most doubly-exponential time, which matches the complexity of standard LTL synthesis. 
Our techniques are applicable to all $\omega$-regular word languages.  For 
specifications given as deterministic Rabin word automata, the complexity is 
polynomial in the number of states and exponential in the number of \changedOne{acceptance condition}
pairs.

\textbf{Outline.} The presentation of this paper is structured as follows.  
Section~\ref{sec:rel} reviews related work and Section~\ref{sec:prel} introduces 
background and notation. Section~\ref{sec:hie} then presents our hierarchy of 
cooperation levels, while Section~\ref{sec:synt} describes our synthesis approach for 
obtaining cooperative systems. We conclude in Section~\ref{sec:concl}. 
\extendedText{
}{
An extended version~\cite{extended} of this paper contains extensions and more  examples.
}

\section{Related Work}\label{sec:rel}

In earlier work~\cite{DBLP:journals/corr/BloemEJK14}, we already pointed out 
that existing synthesis approaches handle environment assumptions 
in an unsatisfactory way.  We postulated that synthesized systems should (1) \emph{be correct}, (2) \emph{not be lazy} by 
satisfying guarantees even if assumptions are violated, (3) \emph{never give up} 
by working towards the satisfaction of the guarantees even if this is not possible in the 
worst case, and (4) \emph{cooperate} by helping the environment to satisfy 
the assumptions that we made about it~\cite{DBLP:journals/corr/BloemEJK14}.  Our current work addresses all of these challenges.  

Besides correctness, our main focus is on cooperation, which is also addressed 
by Assume-Guarantee Synthesis~\cite{DBLP:conf/tacas/ChatterjeeH07} and synthesis 
under rationality assumptions~\cite{DBLP:conf/tacas/FismanKL10, 
DBLP:conf/vmcai/Chatterjee0FR14, DBLP:conf/stacs/Berwanger07, 
DBLP:conf/csl/BrenguierRS14}.  These approaches do not distinguish 
between system and environment, but synthesize implementations for both.
Each component works under certain assumptions about the
other components not deviating from their synthesized implementations 
arbitrarily.  In contrast, our work handles the system and its environment 
asymmetrically: we prefer the guarantees over the assumptions and prioritize
correctness over supporting the environment.

Similar to existing work on synthesis of 
robust~\cite{DBLP:journals/acta/BloemCGHHJKK14} and 
error-resilient~\cite{DBLP:conf/hybrid/EhlersT14} systems, our synthesis 
approach is also \emph{not lazy}~\cite{DBLP:journals/corr/BloemEJK14} in 
satisfying guarantees.  The reason is that satisfying guarantees is more 
preferable in our cooperation hierarchy than satisfying guarantees only if the 
assumptions are satisfied.  In contrast to the existing work, we do not only 
consider the worst case environment behavior for satisfying guarantees (even if 
assumptions are violated), but also the scenario where guarantees can only be 
satisfied with the help of the environment.  

Like~\cite{DBLP:conf/mfcs/Faella09}, we also address the 
\emph{never give up} challenge~\cite{DBLP:journals/corr/BloemEJK14} because our 
cooperation hierarchy also includes levels on which the system cannot enforce 
the guarantees any more.  Still, our synthesis approach \changedOne{lets the system} satisfy the guarantees for \emph{some} environment behavior whenever that is possible.

In contrast to quantitative synthesis~\cite{DBLP:conf/cav/BloemCHJ09, DBLP:conf/icalp/AlmagorBK13} 
where the synthesized system maximizes a certain pay-off, our approach is purely 
qualitative: a certain level of the hierarchy is either achieved or not. In this way, we not only avoid the computational blow-up induced by operating with numerical data, but also remove the need to assign meaningful quantities (such as probabilities) to the specifications.

Finally, \cite{iros15} presents an extension of a synthesis algorithm for 
so-called GR(1) specifications~\cite{DBLP:journals/jcss/BloemJPPS12} to 
produce mission plans for robots that always offer some execution on which both 
assumptions and guarantees are satisfied.  Hence, \cite{iros15} gives a 
solution for one level in our hierarchy, implemented for a subset of LTL.

\section{Preliminaries}\label{sec:prel}

We denote the Boolean domain by $\B=\{\FALSE,\TRUE\}$ and the set of natural 
numbers (including $0$) by $\N$.

\smallskip
\noindent
\textit{Words:} 
We consider synthesis of reactive systems with a finite set 
$\inputs=\{i_1,\ldots,i_m\}$ of Boolean input signals and a finite set 
$\outputs=\{o_1,\ldots,o_n\}$ of Boolean outputs.  The input alphabet is 
$\inalph=2^\inputs$, the output alphabet is $\outalph=2^O$, and $\alp=\inalph 
\times \outalph$.
The set of finite (infinite) words over $\alp$ is denoted by 
$\alp^*$ ($\alp^\omega$), where $\epsilon$ is the word of length $0$.
A set $\lang \subseteq  \alp^\omega$ of infinite words is called a 
(\newterm{word}) \newterm{language}.

\smallskip
\noindent
\textit{Specifications:}
A specification $\varphi$ over $\alp$ defines a language $\lang(\varphi)$ of 
allowed words.  In this paper, specifications consist of two parts, the 
\newterm{environment assumptions} $\Ass$ and the \newterm{system guarantees} 
$\Gua$. These parts can be combined by logical operators. 
For example, we write $\Ass \rightarrow 
\Gua$ to specify that the guarantees (only) need to hold 
if the assumptions are satisfied.
We say that some finite word $w \in \Sigma^*$ is a \newterm{bad prefix} for $\varphi$ if there does not exist a word $w' \in \Sigma^\omega$ such that $ww' \in L(\varphi)$. \changedOne{The set of infinite words that have no bad prefixes of $\varphi$ will be called the \newterm{safety hull} of $\varphi$.}

\smallskip
\noindent
\textit{Reactive systems:}
A reactive system interacts with its environment in a synchronous way. 
In every time step $j$, the system first provides an output letter~$y_j \in \outalph$, 
after which the environment responds with an input letter $x_j \in \inalph$. This 
is repeated in an infinite execution to produce the \newterm{trace} $w = (x_0, 
y_0) (x_1, y_1) \ldots \in \alp^\omega$.  

We can represent the complete behavior of a reactive system by a \newterm{computation tree} $\langle T, \tau \rangle$, where $T$ is the set of nodes and a subset of $\inalph^*$, and $\tau$ assigns labels to the tree nodes. In \changedOne{such a tree}, we have $\tau : T \rightarrow \outalph$, i.e., the tree nodes are labeled by the last output of the system. We call trees $\langle T, \tau \rangle$ with $T = \mathcal{I}^*$ \newterm{full trees} and consider only these henceforth, unless otherwise stated.
We say that some trace $w = (x_0, y_0) (x_1, y_1) \ldots \in \alp^\omega$ is \newterm{included} in $\langle T, \tau \rangle$ if for all $i \in \NN$, $y_i = \tau(x_0 \ldots x_{i-1})$. In such a case, we also say that $w$ is a \newterm{branch} of $\langle T, \tau \rangle$.

We say that some specification 
$\varphi$ is \newterm{realizable} if there exists a computation tree for the system all of whose traces are in $L(\varphi)$. The set of 
traces of $\langle T, \tau \rangle$ is denoted by $\lang(\langle T, \tau \rangle)$ and called the \newterm{word language} of 
$\langle T, \tau \rangle$. The set of all computation trees over $\inalph$ and $\outalph$ is denoted by $\Sysset$. Since a computation tree represents the full behavior of a reactive system, we use these two terms interchangeably.

\smallskip
\noindent
\textit{Rabin word and tree automata:}
In order to represent specifications over $\Sigma$ in a finitary way, we employ \newterm{Rabin automata}. For word languages, we use deterministic Rabin word automata, which are defined as tuples $\changedTwo{\mathcal{R}} = (Q,\Sigma,\delta,q_0,\mathcal{F})$, where $Q$ is the finite set of \newterm{states}, $\Sigma$ is the finite alphabet, $\delta : Q \times \Sigma \rightarrow Q$ is the \newterm{transition function}, $q_0 \in Q$ is the initial state of $\changedTwo{\mathcal{R}}$, and $\mathcal{F} \subseteq 2^Q \times 2^Q$ is the \newterm{acceptance} condition of $\changedTwo{\mathcal{R}}$.
Given a word $w \in \Sigma^\omega$, $\changedTwo{\mathcal{R}}$ induces a \newterm{run} $\pi = \pi_0 \pi_1 \pi_2 \ldots \in Q^\omega$, where $\pi_0 = q_0$ and for every $i \in \NN$, we have $\pi_{i+1} = \delta(\pi_i,w_i)$. 
The run $\pi$ is \newterm{accepting} if there exists some \newterm{acceptance condition pair} $(F,G) \in \mathcal{F}$ such that $\inf(\pi) \cap F = \emptyset$, and $\inf(\pi) \cap G \neq \emptyset$, where $\inf$ maps a sequence to the set of elements that appear infinitely often in the sequence.
We say that $w$ is accepted by $\changedTwo{\mathcal{R}}$ if there exists a run for $w$ that is accepting. The language of $\changedTwo{\mathcal{R}}$, denoted as $L(\changedTwo{\mathcal{R}})$, is defined to be the set of words accepted by $\changedTwo{\mathcal{R}}$. 
The language of a state $q \in Q$ is defined to be the language of the automaton $\changedTwo{\mathcal{R}}_q = (Q,\Sigma,\delta,q,\mathcal{F})$, which differs from $\changedTwo{\mathcal{R}}$ only by its initial state.
Without loss of generality, we assume that every Rabin automaton has a designated state $\top$ that has the full language, i.e., for which $L(\changedTwo{\mathcal{R}}_\top) = \Sigma^\omega$.

In addition to deterministic Rabin word automata, we will later also need non-deterministic Rabin tree automata. A \changedOne{non-deterministic} Rabin tree automaton $\changedTwo{\mathcal{R}} = (Q,\inalph, \allowbreak \outalph, \allowbreak \delta, \allowbreak q_0,\mathcal{F})$ is defined similarly to a word automaton, except that $\delta$ has a different structure, and we now have two alphabets, namely a \newterm{branching alphabet} $\inalph$ and a \newterm{label alphabet} $\outalph$. The transition function $\delta$ is defined as $\delta : Q \times \outalph \rightarrow 2^{\inalph \rightarrow Q}$.
A tree automaton accepts or rejects computation trees instead of words. 
Given a computation tree $\langle T, \tau \rangle$, we say that some tree $\langle T', \tau' \rangle$ is a \newterm{run tree} of $\changedTwo{\mathcal{R}}$ for $\langle T, \tau \rangle$ if $T' = T$, $\tau' : T' \rightarrow Q$, $\tau'(\epsilon)=q_0$, and for all $t' \in T'$, there exists a function $f \in \delta(\tau'(t'),\tau(t'))$ such that for all $i \in \mathcal{I}$, $f(i) = \tau'(t' i)$. We say that $\langle T', \tau' \rangle$ is \newterm{accepting} if every branch of $\langle T', \tau' \rangle$ is accepting, i.e., its sequence of labellings satisfies the Rabin acceptance condition $\mathcal{F}$.
We say that a tree automaton accepts a computation tree if it has a corresponding accepting run tree.

Given a Rabin word automaton $\changedTwo{\mathcal{R}}$ as specification over the alphabet $\Sigma = \mathcal{I} \times \mathcal{O}$, we can check if there exists a tree with branching alphabet $\mathcal{I}$ and label alphabet $\mathcal{O}$ all of whose traces are in the language of $\changedTwo{\mathcal{R}}$.
It has been shown that this operation can be performed in time exponential in $|\mathcal{F}|$ and polynomial in $|Q|$ by first translating the Rabin word automaton to a non-deterministic Rabin tree automaton with the same set of states and the same acceptance condition (in linear time), and then checking the tree automaton's language for emptiness \cite{DBLP:conf/popl/PnueliR89}.
This approach gives rise to a reactive synthesis procedure for specifications in \newterm{linear temporal logic} (LTL) \cite{DBLP:conf/popl/PnueliR89}: we can first translate the LTL specification to a deterministic Rabin word automaton with a number of states that is doubly-exponential in the length of the specification, and a number of acceptance condition pairs that is exponential in the length of the specification.
Overall, this gives a doubly-exponential time procedure for LTL reactive synthesis, which matches the known complexity of the problem.



\section{A Hierarchy of Cooperation Levels}\label{sec:hie}

In this section, we develop a hierarchy of cooperation levels that a system may 
achieve. 
We use a special logic to specify desired properties of the system to be 
synthesized.  The syntax of a formula $\Phi$ in this logic is defined as
\begin{equation}
\Phi ::= \varphi \mid 
           \reach \varphi \mid
           \reachinf \varphi \mid
           \Phi \wedge \Phi,
\label{eqn:cooperationLevelGrammar}
\end{equation}
where $\varphi$ is a linear-time specification (e.g., an LTL formula) over an alphabet $\alp$.  The Boolean connective $\wedge$ 
has its expected semantics. The formula $\reach \varphi$ is satisfied for a 
system $\Sys$ if there exists some trace of $\Sys$ on which $\varphi$ holds. 
Similarly, the formula $\reachinf \varphi$ is satisfied if, from any 
point in a reactive system's trace that has been seen so far, there exists some suffix trace of the system
such that $\varphi$ holds. 
We call an instance of the grammar in Formula~\ref{eqn:cooperationLevelGrammar} a \newterm{cooperation level specification}.
Our logic ranges over computation trees and has similarities to \newterm{strategy logic}~\cite{DBLP:journals/iandc/ChatterjeeHP10} as well as 
\emph{alternating-time temporal logic} (ATL)~\cite{DBLP:journals/jacm/AlurHK02}. 
Yet, its semantics, to be given below, is very different. In particular, linear-time specifications $\varphi$ are seen as atomic and even within the scope of a $\LTLG$ operator, $\varphi$ is only  evaluated on \emph{complete traces} of a system, always starting \changedOne{at} the root of a computation tree. 

More formally, we define the semantics as follows.  Formulas are interpreted 
over a computation tree $ \langle T, \tau \rangle \in\Sysset$ using the
following rules:
{\allowdisplaybreaks
\label{pageref:cooperationLevelSemantics}
\begin{align*}
&\langle T, \tau \rangle  &\models& \varphi &\;\;\text{ iff }\;\;&
   \forall w \in L(\langle T, \tau \rangle): w \in \lang(\varphi)\\
&\langle T, \tau \rangle  &\models& \reach \varphi &\;\;\text{ iff }\;\;& 
   \exists w \in L(\langle T, \tau \rangle): w\in \lang(\varphi)\\
&\langle T, \tau \rangle  &\models& \reachinf \varphi &\;\;\text{ iff }\;\;& 
   \forall w \in L(\langle T, \tau \rangle ), i \in \NN \scope
   \exists w' \in L(\langle T, \tau \rangle ) \cap L(\varphi)\scope \\
   & & & & & w_0 \ldots w_{i-1} = w'_{0} \ldots w'_{i-1} \\
& \langle T, \tau \rangle  &\models& \Phi_1 \wedge \Phi_2 &\;\;\text{ iff }\;\;&  
   (\langle T, \tau \rangle  \models \Phi_1) \wedge (\langle T, \tau \rangle  \models \Phi_2)
\end{align*}
}

\noindent
Most cases are straightforward, but $\reachinf \varphi$ requires some attention.  
We follow an arbitrary trace of $\langle T, \tau \rangle$ up to step $i-1$. At step 
$i-1$, there must be some continuation of the trace that satisfies $\varphi$ and that is part of the system's computation tree. If such a continuation 
exists in every step $i$, then $\reachinf \varphi$ is satisfied. 

\subsection{Defining the Interesting Levels of Cooperation} \label{sec:defcop}

We start with the linear-time specification $\mathcal{A}$, which represents the assumptions about the environment, and $\mathcal{G}$, which represents the guarantees. 
We can combine them to the linear-time properties $\mathcal{A} \rightarrow \mathcal{G}$ and $\mathcal{A} \wedge \mathcal{G}$. 
The first of these represents the classical correctness requirement for reactive synthesis, while the latter represents the optimistic linear-time property that both assumptions and guarantees hold along a trace of the system.
As generators for cooperation level specifications, we consider all rules from the grammar in Equation~\ref{eqn:cooperationLevelGrammar} except for the second one, as it is rather weak, and only considers what can happen from the initial state of a system onwards.
Additionally, leaving out the conjuncts of the form $\langle E \rangle \varphi$ strengthens the semantic foundation for our \newterm{maximally cooperative synthesis} approach in Sect.~\ref{sec:optimisticLevelJumping}. However, we discuss the consideration of conjuncts of the form $\langle E \rangle \varphi$ in Section~\ref{sec:alt}.

Combining all four linear-time properties with the two chosen ways of lifting a linear-time property to our logic gives the following different conjuncts for cooperation level specifications:
\begin{equation}
\label{eqn:cooperationLevelConjunctsConsidered}
D = \{\Ass \rightarrow \Gua, \Gua, \Ass,  \reachinf(\Ass \wedge \Gua), 
  \reachinf\Gua, \reachinf\Ass, \reachinf(\Ass \rightarrow \Gua)\}
\end{equation}
A \newterm{cooperation level specification} is a conjunction between elements from this set. So there are $2^7 = 128$ possible cooperation levels in this setting. Note that we removed the linear-time property $\mathcal{A} \wedge \mathcal{G}$ from $D$, as it can be simulated by a conjunction between $\mathcal{A}$ and $\mathcal{G}$ on the level of cooperation level specifications.

We can reduce the $128$ possible cooperation levels substantially by using our knowledge of the semantics of cooperation level specifications, which can be expressed in \newterm{reduction rules}.
For example, if in a cooperation level specification, for some linear-time property $\varphi \in \{\mathcal{A}, \mathcal{G}, \mathcal{A} \rightarrow \mathcal{G}, \mathcal{A} \wedge \mathcal{G}\}$, $\varphi$ is a conjunct along with $\LTLG \langle E \rangle \varphi$, then we can remove the latter from the cooperation level specification as that conjunct is implied by the former. This is because if along every trace of a computation tree, $\varphi$ holds, then for every node in the tree, we can find a trace containing the node along with $\varphi$ holds. In a similar fashion, we can observe  
\begin{itemize}
\item that $\mathcal{G}$ implies $\mathcal{A} \rightarrow \mathcal{G}$, 
\item that $\reachinf(\mathcal{A} \wedge \mathcal{G})$ implies $\reachinf\mathcal{A}$ and $\reachinf\mathcal{G}$,
\item that $\reachinf\mathcal{G}$ implies $\reachinf{(\mathcal{A} \rightarrow \mathcal{G})}$,
\item that $\mathcal{A} \rightarrow \mathcal{G}$ and $\mathcal{A}$ together imply $\mathcal{G}$, and
\item that $\mathcal{A} \rightarrow \mathcal{G}$ and $\reachinf \mathcal{A}$ together imply $\reachinf(\mathcal{A} \wedge \mathcal{G})$,
\item that $\mathcal{A}$ and $\reachinf \mathcal{G}$ together imply $\reachinf(\mathcal{A} \wedge \mathcal{G})$,
\item that $\reachinf(\mathcal{A} \rightarrow \mathcal{G})$ and $\mathcal{A}$ together imply $\reachinf(\mathcal{G})$.
\end{itemize}
Only in the last \changedOne{three} of these \newterm{rules}, conjuncts of the forms $\reachinf\varphi$ and $\varphi$ interact. For example, if $\mathcal{A} \rightarrow \mathcal{G}$ holds along all traces of a computation tree, and we know that through every node in the tree, there is a trace on which $\mathcal{A}$ holds, then along this trace, $\mathcal{A} \wedge \mathcal{G}$ holds as well. Thus, we also know that the computation tree satisfies $\reachinf(\mathcal{A} \wedge \mathcal{G})$. Note that $\reachinf(\Ass \wedge \Gua)$ is not 
equal to $\reachinf(\Gua) \wedge \reachinf(\Ass)$ because there exist computation trees for which a part of their traces satisfy $\Ass\wedge\neg\Gua$, the other traces satisfy $\Gua\wedge\neg\Ass$, but none of their traces satisfy $\Ass\wedge\Gua$. 

\begin{figure}[tb]
\centering
\begin{tikzpicture}[auto]
\node[rectangle,draw,fill=gray!30] at (0,0)    (S0) {$\Gua \wedge \Ass$};
\node[rectangle,draw]              at (-2,-1)  (S11){$\Ass \wedge \reachinf \Gua$};
\node[rectangle,draw,fill=gray!30] at (2,-1)   (S12){$\Gua \wedge \reachinf \Ass$};

\node[rectangle,draw]              at (-3.5,-2)(S21){$\Ass$};
\node[rectangle,draw,fill=gray!30] at (2,-2)   (S22){$(\Ass \rightarrow \Gua) \wedge \reachinf \Ass$};
\node[rectangle,draw,fill=gray!30] at (4.5,-2) (S23){$\Gua$};

\node[rectangle,draw]              at (0,-3)   (S31){$\reachinf(\Ass \wedge \Gua)$};
\node[rectangle,draw,fill=gray!30] at (4.5,-3) (S32){$(\Ass \rightarrow \Gua) \wedge \reachinf \Gua$};

\node[rectangle,draw]              at (0,-4)   (S4) {$(\reachinf\Ass) \wedge (\reachinf\Gua)$};

\node[rectangle,draw]              at (-1,-5)  (S5) {$(\reachinf \Ass) \wedge \reachinf(\Ass \rightarrow \Gua)$};

\node[rectangle,draw]              at (-3.5,-6)(S61){$\reachinf \Ass $};
\node[rectangle,draw]              at (4.5,-6) (S62){$\reachinf \Gua $};
\node[rectangle,draw,fill=gray!30] at (6.5,-6) (S63){$\Ass\rightarrow\Gua$};
\node[rectangle,draw]              at (3.5,-7) (S7) {$\reachinf(\Ass \rightarrow \Gua)$};

\path
(S0) edge (S11)
(S0) edge (S12)
(S11) edge (S21)
(S12) edge (S22)
(S12) edge (S23)
(S11) edge (S31)
(S22) edge (S31)
(S22) edge (S32)
(S23) edge (S32)
(S31) edge (S4)
(S4) edge (S5)
(S5) edge (S61)
(S21) edge (S61)
(S32) edge (S62)
(S32) edge (S63)
(S4) edge (S62)
(S5) edge (S7)
(S62) edge (S7)
(S63) edge (S7)
;
\end{tikzpicture}
\caption{A Hasse diagram of the hierarchy of cooperation levels.}
\label{fig:hierarchy}
\end{figure}
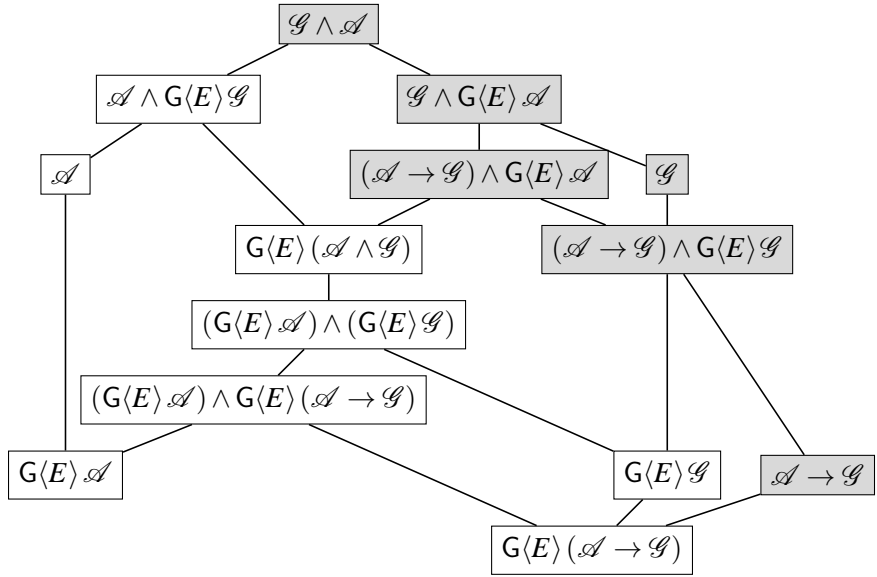

After reducing the set of distinct cooperation levels by these reduction rules, we obtain only $15$ semantically different subsets of $D$, which form a lattice with its partial order defined by implication.
We leave out the $\TRUE$ element of the lattice in the following (as it is trivially satisfied by all computation trees), and visualize the remaining cooperation levels hierarchically in Fig.~\ref{fig:hierarchy}. For every cooperation level, implied elements of $D$ have also been left out in the vertex labeling. 
Every edge in the 
figure denotes a ``stricter than''-relation, with the upper element being the stricter one.
Vertices that enforce the traditional correctness criterion $\Ass 
\rightarrow \Gua$ are colored in gray.

In the next section, we discuss the different levels of the cooperation 
hierarchy in an example.  After that, we will discuss design decisions that we made when constructing the hierarchy, as well as extensions and alternatives.

\subsection{Example} \label{sec:ex}

This section presents an example from real life to illustrate some 
cooperation levels from Fig.~\ref{fig:hierarchy}.  
\extendedText{
A more extensive example will be presented in Appendix~\ref{sec:app:ex}.
}{
A more extensive example is given in the extended version of this paper~\cite{atva15}.
}

An IT provider makes a contract with a customer.
The provider promises that whenever the computer breaks, it will eventually
deliver a new computer, or one of its technicians will come to the customer to fix the computer while providing some free printing supplies as a bonus.
The customer accepts the conditions, mainly because of the possibility of free printing supplies. 
The contract between the customer and the IT provider contains the following assumptions under which it has to be fullfilled: (1) every traffic blockade on the way to the \changedTwo{customer} is eventually cleared, and (2) if the road to the customer is clear and the IT provider technician is working on the computer, then it will eventually be fixed. The latter assumption is part of the contract to exclude problems with the customer's software, which is not the responsibility of the IT provider. By simply replacing parts of the computer one-by-one, the problem \changedTwo{is} eventually fixed. We also have an additional assumption that traffic can be blocked on purpose by the IT provider, which is not written into the contract.

We model the setting with two Boolean input variables $c,b$ and four Boolean output variables
$t,d,m,n$ for the IT service provider's behavior.  Variable $c$ is $\TRUE$ if the computer is currently working, $b$ is 
$\TRUE$ when the traffic is not blocked, $f=\TRUE$ indicates that the IT provider's technician is trying to fix 
the computer, $m=\TRUE$ means that free printing supplies are delivered, $n=\TRUE$ means 
that the IT service provider delivers a new computer, and $d=\TRUE$ means that the provider currently blocks the traffic on purpose. In LTL syntax\footnote{In LTL syntax, $\LTLG$ means ``always'' and $\LTLF$ means 
``eventually''.}, the guarantee can be written as 
$\Gua = \LTLG\bigl(\neg c \rightarrow \LTLF(n \vee (c \wedge m))\bigr)$. 
The assumptions can be formalized as 
$
\Ass =
\LTLG(\LTLF b ) \wedge 
\LTLG\bigl((b \wedge f) \rightarrow \LTLF c\bigr) \wedge
\LTLG(d \rightarrow \neg b)
$.

The following table summarizes some cooperation levels that can be achieved 
with different behavior of the IT provider, each expressed as an LTL property whose fulfillment completely determines the valuation of the output variables in all time steps.
We focus on levels that enforce $\Ass \rightarrow 
\Gua$ (which are colored gray in Fig.~\ref{fig:hierarchy}).
\begin{center}
\setlength{\tabcolsep}{20pt}
\begin{tabular}{c|c|c}
Nr.  &Behavior &Level in Fig.~\ref{fig:hierarchy}  \\
\hline
1    & $\LTLG(d \wedge \neg f \wedge \neg m \wedge \neg n)$
     &$\Ass \rightarrow \Gua$
\\
2    & $\LTLG(d \wedge m \wedge \neg f \wedge \neg n)$
     &$(\Ass \rightarrow \Gua) \wedge \reachinf \Gua$
\\
3    & $\LTLG(f \wedge m \wedge \neg n \wedge \neg d)$
     &$(\Ass \rightarrow \Gua) \wedge \reachinf \Ass$ 
\\
4    & $\LTLG(n \wedge d \wedge \neg f \wedge \neg m)$
     &$\Gua$ 
\\
5    & $\LTLG(n \wedge \neg d \wedge \neg f \wedge \neg m)$
     &$\Gua \wedge \reachinf \Ass$ 
\\
\end{tabular}
\end{center}

Behavior~1 enforces an assumption violation by blocking the traffic.  This is very 
uncooperative with the customers (and all other drivers on the streets). In this case,
$\reachinf \Ass$ does not hold as the environment cannot satisfy $\Ass$ from any 
point in any trace because it cannot set $b$ to $\TRUE$ at some point without violating $\mathcal{A}$, but not ever doing so violates $\Ass$ as well. The cooperation level specification part $\reachinf \Gua$ does not hold either 
because since $n$ and $m$ are both $\FALSE$ all of the time, $\Gua$ is not fulfilled along any trace.

Behavior~2 is better because printing supplies are always delivered. This satisfies $\reachinf \Gua$ 
because, at any point, the environment can make the computer work again (which occasionally happens with  computers).  $\reachinf \Ass$ is still not satisfied for the same 
reason as before.

Behavior~3 is even less destructive.  The technician tries to fix the computer 
($f=\TRUE$) and brings along free printing supplies ($m=\TRUE$) without blocking traffic.  This 
satisfies the guarantees if the assumptions are satisfied.  On top of that, the 
assumptions can always be satisfied by setting both $b$ and $c$ to $\TRUE$.  
Since $d$ is always $\FALSE$, the last assumption (which was the problem before) 
always holds.

Behavior~4 is better than Behavior~2 but incomparable with Behavior~3.  By 
always (or repeatedly) providing a new computer, $\Gua$ is enforced independent of 
the assumptions.  However, $\reachinf \Ass$ does not hold because, since the 
traffic is blocked ($d=\TRUE$), the assumptions cannot hold.

Behavior~5 is similar but without enforcing a traffic blockade.  It thus satisfies $\Gua$ 
and $\reachinf \Ass$ simultaneously.  The last two of the four assumptions are even 
enforced. However, since the first assumption cannot be enforced by any behavior of the
IT provider, $\Gua \wedge \Ass$ cannot be achieved.

\subsection{Discussion} \label{sec:alt}

Before solving the \newterm{cooperative synthesis} problem in the next section, let us discuss some interesting aspects of our hierarchy of cooperation levels.

\smallskip
\noindent
\textit{Incomparable levels:}
As already discussed in the example, our hierarchy of cooperation levels 
in Fig.~\ref{fig:hierarchy} contains levels that are incomparable, i.e., 
where neither of the levels is stricter or clearly more 
desirable.  The relative preferences between such levels may depend on the 
application.  Among the levels that enforce $\Ass \rightarrow \Gua$ (gray in 
Fig.~\ref{fig:hierarchy}), there is only one incomparability, namely between 
$\Gua$ and $(\Ass \rightarrow \Gua) \wedge \reachinf \Ass$.  The former favors 
the guarantees, the latter the assumptions.  For incomparabilities between a 
level that enforces $\Ass \rightarrow \Gua$ and one that does not, we suggest 
to prefer the former unless there are good reasons not to.  

\smallskip
\noindent
\textit{Symmetry in Fig.~\ref{fig:hierarchy}:}
Our hierarchy is asymmetric because we included $\Ass \rightarrow \Gua$ but not 
$\Gua \rightarrow \Ass$ in $D$.  Including $\Ass \rightarrow \Gua$ contradicts 
the philosophy of cooperation to some extend, but is justified by the fact that we always take the point of view of the system in this paper, and try to synthesize a \emph{correct implementation} that also helps the environment to satisfy its assumptions whenever this is reasonable, but without assuming that the environment has its own goals and behaves \emph{rationally} (as in \cite{DBLP:conf/tacas/ChatterjeeH07,DBLP:conf/tacas/FismanKL10, 
DBLP:conf/vmcai/Chatterjee0FR14, DBLP:conf/stacs/Berwanger07, 
DBLP:conf/csl/BrenguierRS14}). The guarantees often cannot be enforced unconditionally by the system, so the system has to rely on the assumptions to hold. However, the combination of $\mathcal{A} \rightarrow \mathcal{G}$ with $\reachinf\Ass$ (or $\reachinf\Gua$) in a cooperation level specification eliminates the possibility for the system to achieve correctness by simply enforcing a violation of $\Ass$ (and $\Gua$).

\smallskip
\noindent
\textit{Additional operators:}
We did not include any properties with a plain $\reach$-operator in our default 
hierarchy.  The reason is that $\reach\varphi$ is quite a weak goal. The system 
could have exactly one trace on which $\varphi$ is satisfied.  As soon as the 
environment deviates from the input along this trace, 
the system can behave arbitrarily.  Our default hierarchy also 
does not contain any occurrences of $\Ass \vee \Gua$, but this is mainly to keep 
the presentation simple.  Including $\Ass \vee \Gua$ and $\reachinf(\Ass \vee 
\Gua)$ would extend the hierarchy to $23$ levels.  Additionally including 
$\reach$ wherever $\reachinf$ is applied results in $77$ 
levels. Thus, even with extensions in the applied operators, the size of our hierarchy remains manageable.
Details to these refined hierarchies and the reduction rules to
obtain them can be found  
\extendedText{
in Appendix~\ref{sec:app:hie}.
}{
in the appendix of~\cite{extended}.
}

\smallskip
\noindent
\textit{Fine-grainedness:}
Our hierarchy considers two dimensions: the goals (in terms of assumptions and 
guarantees) to achieve, and the certainty with which the goal can be achieved: 
enforced, always reachable ($\reachinf$), and initially reachable ($\reach$, if 
considered).  In contrast, most existing synthesis 
approaches~\cite{DBLP:journals/acta/BloemCGHHJKK14, DBLP:conf/hybrid/EhlersT14, 
DBLP:conf/cav/BloemCHJ09} (see Section~\ref{sec:rel}) that go beyond plain 
correctness only focus on \emph{enforcing} a maximum in the dimension of goals. 
 Yet, in this dimension they are often more fine-grained, e.g., by considering 
individual assumptions and guarantees or how often some property is violated.  
In principle, we could also \changedTwo{increase} the granularity in our goal dimension. 
However, this also comes at a price: it would increase the size of the hierarchy 
and induce more incomparabilities, which makes it more difficult 
to define the preference between the incomparable levels for a concrete application.

\section{Synthesizing Desirable Systems} \label{sec:synt}
\changedOne{
After defining our hierarchy of cooperation levels, we turn towards the synthesis of implementations that maximize the possible cooperation level.
We start by describing how we can synthesize an implementation for a single cooperation level, and then show in Sect.~\ref{sec:optimisticLevelJumping} how to synthesize \newterm{maximally cooperative implementations}, which move upwards in the cooperation level hierarchy whenever possible during the execution.

\subsection{Implementing a \changedTwo{S}ingle \changedTwo{C}ooperation \changedTwo{L}evel}
\label{sec:singleCooperationLevel}

The simple reduction of the cooperative synthesis problem for LTL to synthesis from a logic such as \emph{strategy logic}~\cite{DBLP:journals/iandc/ChatterjeeHP10} is obstructed by the fact that in our semantics, we always evaluate traces from the start when evaluating a subformula of the shape $\reachinf \varphi$. Strategy logic lacks a \newterm{rewind} operator that would allow to jump back to the start of a trace. We could, however, encode a cooperation level specification into CTL* with linear past \cite{Bozzelli}, and use a corresponding synthesis procedure. 

%
%

Instead of following this route, we give a direct automata-theoretic approach to the synthesis of reactive systems that implement some level of cooperation with the environment on a specification $(\Ass, \Gua)$. The automata built in this way keep information about the states of the assumptions and guarantees explicit, which is needed to synthesize maximally cooperative implementations in the next subsection. Also, it makes the synthesis approach applicable to general $\omega$-regular word language specifications.
}

Starting from the linear-time properties $\Ass$ and $\Gua$, we show how to build a non-deterministic Rabin tree automaton that encodes the synthesis problem for some cooperation level specification $H$. Such a tree automaton can be checked for language emptiness in order to perform synthesis for one level of cooperation. 

%
Let $\mathcal{R}_\Ass$, $\mathcal{R}_\Gua$, $\mathcal{R}_{\Ass \rightarrow \Gua}$, and $\mathcal{R}_{\Ass \wedge \Gua}$ be deterministic Rabin word automata that encode $\Ass$, $\Gua$, $\Ass \rightarrow \Gua$, and $\Ass \wedge \Gua$.
We translate the conjuncts (elements of $D$ from Equation~\ref{eqn:cooperationLevelConjunctsConsidered}) of some expression $H$ in the grammar from Equation~\ref{eqn:cooperationLevelGrammar} to non-deterministic Rabin tree automata $\mathcal{T}^D$ individually and then build the product $\mathcal{T}^H$ between these tree automata, which encodes that all elements of $H$ have to be satisfied in a candidate computation tree. As cooperation level specifications have only three types of conjuncts, namely $\varphi$, $\langle E \rangle \varphi$, and $\LTLG \langle E \rangle \varphi$ for some linear-time property $\varphi$, we simply give the translations for these types separately.

\textbf{Case $\varphi$:} If $\varphi$ is a linear-time property, the occurrence of $\varphi$ in a cooperation level specification indicates that $\varphi$ should hold on every trace of a synthesized implementation. Translating $\mathcal{R}_\varphi = (Q,\inalph \times \outalph,\delta,q_0,\mathcal{F})$ to a Rabin tree automaton $\mathcal{T}^D = (Q,\inalph,\outalph,\delta',q_0,\mathcal{F})$ that enforces $\varphi$ to hold along all traces in the tree is a standard construction, where we set:
\begin{align*}
\forall q \in Q, o \in \outalph: \delta'(q,o) & = \{ \{ i \mapsto \delta(q,(i,o)) \mid i \in \inalph \} \}
\end{align*}

\textbf{Case $\langle E \rangle \varphi$:}
Here, we require the synthesized system to offer one path along which $\varphi$ holds. In the tree automaton, we non-deterministically choose this path. Starting with $\mathcal{R}_\varphi = (Q,\inalph \times \outalph,\delta,q_0,\mathcal{F})$, we obtain $\mathcal{T}^D = (Q,\inalph,\outalph,\delta',q_0,\mathcal{F})$ with:
\begin{align*}
\forall q \in Q, o \in \outalph: \delta'(q,o) & = \bigcup_{i \in \inalph} \{ \{ i \mapsto \delta(q,(i,o)) \} \cup \{ i' \mapsto \top \mid i' \in \inalph, i' \neq i \} \}
\end{align*}

\textbf{Case $\LTLG \langle E \rangle \varphi$:}
In this case, for every node in a computation tree, we require the synthesized system to have a path on which $\varphi$ holds that includes the selected node. Given $\mathcal{R}_\varphi = (Q,\inalph \times \outalph,\delta,q_0,\mathcal{F})$, we can implement this requirement as a non-deterministic Rabin tree automaton $\mathcal{T}^D = (Q',\inalph,\outalph,\delta',q'_0,\mathcal{F}')$ with:
\begin{align*}
Q' & = Q \times \BB \\
\forall (q,b) \in Q', o \in \outalph: \delta'((q,b),o) & = \bigcup_{i \in \inalph} \{ \{ i \mapsto (\delta(q,(i,o)),\TRUE) \} \\
& \quad \quad \;\; \cup \{ i' \mapsto (\delta(q,(i',o)),\FALSE) \mid i' \in \inalph, i' \neq i \} \} \\
q'_0 & = (q_0,\TRUE) \\
\mathcal{F}' & = \{ (F \times \{\TRUE\}, G \times \{\TRUE\}) \mid (F,G) \in \mathcal{F} \} \\ & \quad \quad \cup \{ (\emptyset, Q \times \{\FALSE\}) \}
\end{align*}
The automaton augments the states in $Q$ by a Boolean flag. From every node in a computation tree accepted by $\mathcal{T}^D$, regardless of whether it is flagged by $\TRUE$ or $\FALSE$, there must exist a branch consisting only of $\TRUE$-labeled nodes. The original acceptance condition $\mathcal{F}$ must hold along this branch. However, not all branches of a tree accepted by $\mathcal{T}^D$ have to satisfy $\varphi$, as those branches along which the flag is $\FALSE$ infinitely often are trivially accepting. Intuitively, $\mathcal{T}^D$ also enforces the safety hull of $\varphi$ along all branches in the computation tree as once a bad prefix has been seen on the way to a computation tree node, there cannot exist a branch containing the node along which $\varphi$ holds, which is enforced by the (required) successor branch that is always labeled by $\TRUE$.

In order to obtain a non-deterministic Rabin tree automaton $\mathcal{T}^H$ for a complete cooperation level specification $H$, we can compute a product automaton from the Rabin tree automata $\mathcal{T}^D$ for the individual conjuncts. Computing such a product automaton is a standard operation in automata theory. Given a set $\{ (Q_1,\mathcal{I},\mathcal{O},\delta_1,q_{0,1},\mathcal{F}_1), \ldots,$ $(Q_n,\mathcal{I}, \allowbreak \mathcal{O}, \allowbreak \delta_n,q_{0,n}, \allowbreak \mathcal{F}_n) \}$ of non-deterministic Rabin tree automata, their product is defined as the non-deterministic Rabin tree automaton $\mathcal{T}^H = (Q,\mathcal{I},\mathcal{O},\delta,q_0,\mathcal{F})$ with:
{\allowdisplaybreaks
\begin{align*}
Q & = Q_1 \times \ldots \times Q_n \\
\delta((q_1, \ldots, q_n),o) & = \bigotimes_{i \in \{1, \ldots, n\}} \delta_i(q_i,o) \\
q_0 & = (q_{0,1}, \ldots, q_{0,n}) \\
\mathcal{F} & = \{ (F_1 \times \ldots \times F_n, G_1 \times \ldots \times G_n) \mid \\
&\hspace*{0.7cm} (F_1,G_1) \in \mathcal{F}_1, \ldots,  (F_n,G_n) \in \mathcal{F}_n \}
\end{align*}
}

\noindent
The second line of this equation holds for all 
$(q_1, \ldots, q_n) \in Q$ and $o \in \mathcal{O}$.
Also, we used the special operator $\bigotimes$ that maps sets $\{M_i \subseteq 2^{Q_i \rightarrow \mathcal{I}} \}_{j \in \{1, \ldots, n\}}$ to the set
\begin{equation*}
\{ \{ ((q_1, \ldots, q_n),i) \in Q \times \mathcal{I} \mid \bigwedge_{j \in \{1, \ldots, n\}} f_j(i) = q_j \} \mid f_1 \in M_1, \ldots, f_n \in M_n \}
\end{equation*}
By performing reactive synthesis using the product Rabin tree automaton as specification, we can obtain an implementation that falls into the chosen cooperation level. 

\subsection{Maximally Cooperative Synthesis}
\label{sec:optimisticLevelJumping}
For linear-time specifications $\mathcal{A}$ and $\mathcal{G}$, there are typically some cooperation level specifications that are not realizable (such as $\mathcal{A} \wedge \mathcal{G}$) and some that are realizable (such as $\LTLG \langle E \rangle (\mathcal{A} \rightarrow \mathcal{G})$). 
By iterating over all cooperation level specifications and applying reactive synthesis for the product Rabin tree automat\changedOne{a} $\mathcal{T}'$ computed by the construction from above, we can check which is the highest cooperation level that can be realized and compute a respective~implementation.

In some cases, there may however be the possibility to switch to a higher level of cooperation during the system's execution. Take for example the case that initially, the cooperation level $\LTLG \langle E \rangle \mathcal{A} \wedge (\mathcal{A} \rightarrow \mathcal{G})$ is not realizable, but $(\mathcal{A} \rightarrow \mathcal{G})$ is. 
This may be the case if $\mathcal{A}$ represents the constraint that the environment has to perform a certain finite sequence of actions in reaction to the system's output, which is representable as a safety property. If the actions are triggered by the system, then the system cannot ensure that the environment succeeds with performing them correctly, hence violating $\mathcal{A}$.
If $\mathcal{A} \rightarrow \mathcal{G}$ can \changedOne{only} be realized by the system by triggering these actions \changedOne{at least once}, then after these actions have been performed by the environment, the cooperation level specification $\LTLG \langle E \rangle \mathcal{A} \wedge (\mathcal{A} \rightarrow \mathcal{G})$ can however be enforced by the system by not triggering them again.

This observation motivates the search for \newterm{maximally cooperative implementations}, which at any point in time realize the highest possible cooperation level. 
Before describing how to synthesize such implementations, let us first formally define what this~means.






When determining the cooperation level during the execution of a system, we only look at the part of its computation tree that is consistent with the input obtained from the environment so far. 
Given a computation tree $\langle T, \tau \rangle$ and the input part of a prefix trace $w^\mathcal{I} = w^\mathcal{I}_0 w^\mathcal{I}_1 \ldots w^\mathcal{I}_n\in \mathcal{I}^*$, we define the \newterm{bobble tree} \cite{DBLP:conf/rv/EhlersF11} of $\langle T, \tau \rangle$ for $w^\mathcal{I}$ to be the tree $\langle T', \tau' \rangle$, where $T' = \{ \epsilon \} \cup \{ w^\mathcal{I}_0 w^\mathcal{I}_1 \ldots w^\mathcal{I}_k \mid k \leq n \} \cup \{ \changedOne{w^\mathcal{I}} t \mid t \in \mathcal{I}^* \}$, and $\tau'(t) = \tau(t)$ for all $t \in T'$. 
We call $w^\mathcal{I}$ the \newterm{split node} of $\langle T', \tau' \rangle$ and $(\tau(\epsilon),w^\mathcal{I}_0) (\tau(w^\mathcal{I}_0),w^\mathcal{I}_1) \ldots (\tau(w^\mathcal{I}_0 \ldots w^\mathcal{I}_{n-1}), \allowbreak w^\mathcal{I}_n)$ the \newterm{split word} of $\langle T', \tau' \rangle$.
Intuitively, the bobble tree has a single path to the split node
$w^\mathcal{I}$ and is full from that point onwards.
Cutting a full computation tree into a bobble tree does not reduce the cooperation level that the tree fulfills for the specification types in the classes that we can built from the conjuncts in $D$ (from Eqn.~\ref{eqn:cooperationLevelConjunctsConsidered}):
\begin{lemma}
\label{lem:bobbleTreeSplitting}
Let $\langle T, \tau \rangle$ be a computation tree and $\langle T', \tau' \rangle$ be a bobble tree built from $\langle T, \tau \rangle$. If for some cooperation level specification $H$ consisting of conjuncts in $D$, we have that $\langle T, \tau \rangle$ fulfills $H$, then $\langle T', \tau' \rangle$ also fulfills $H$.
\end{lemma}
\begin{proof}
Proof by induction over the structure of $H$, using the semantics given on page~\pageref{pageref:cooperationLevelSemantics}. All conjuncts in $D$ have an outermost universal quantification over the elements in $L(\langle T, \tau \rangle)$. Reducing the number of elements in $L(\langle T, \tau \rangle)$ does not make these constraints harder to fulfill.
\end{proof}
Bobble trees provide us with a semantical basis for switching between cooperation levels: if a reactive system $\langle T, \tau \rangle$ for a cooperation level $H$ executes, and after some prefix trace $w$, there exists a bobble tree with split word $w$ that allows a strictly higher cooperation level $H'$, then it makes sense to continue the execution of the system \changedOne{according to cooperation level $H'$}. We thus define:

\begin{definition}
Let $(\mathcal{A}, \mathcal{G})$ be a linear-time specification. We call a computation~tree $\langle T, \tau \rangle$ \newterm{maximally cooperative} if for every split node $t \in T$, the bobble tree induced by $t$ and $\langle T, \tau \rangle$ implements a highest possible cooperation level for $\mathcal{A}$ and $\mathcal{G}$ among the bobble trees with the same split word.
\end{definition}
Note that since for some bobble tree, there may be multiple highest cooperation levels, it makes sense to define a preference order for the cooperation level specifications, so that when synthesizing an implementation, the implementation can always pick the most desired one when the possibility to move up in the hierarchy arises.

In order to synthesize maximally cooperative implementations, we first need a way to check, for every split word, for the existence of a bobble tree for cooperation level specifications.
The special structure of the tree automata $\mathcal{T}^H$ built according to the construction in Sect.~\ref{sec:singleCooperationLevel} offers such a way. 

\begin{definition}
Let $\mathcal{T}^H = (Q,\inalph,\outalph,\delta,q_0,\mathcal{F})$ be a non-deterministic tree automaton for a cooperation level specification built  according to the (product) construction from Sect.~\ref{sec:singleCooperationLevel}. We have that $Q$ is of the shape $C_1 \times \ldots \times C_n$, where for every $i \in \{1, \ldots, n\}$, we either have $C_i = Q'$ or $C_i = Q' \times \BB$ for some Rabin word automaton state set $Q'$. For a state $q = (c_1, \ldots, c_n) \in Q$, we define $\mathit{unpack}(q) = \mathit{unpack}(c_1) \cup \ldots \cup \mathit{unpack}(c_n)$, where we concretize $\mathit{unpack}(q') = q'$ and $\mathit{unpack}((q',b))=q'$ for some word automaton state $q'\in Q'$ and $b \in \BB$.
\end{definition}

\begin{lemma}
\label{lem:productTreeAutomatonProperties}
Let $\mathcal{T}^H$ be a tree automaton buil\changedTwo{t} according to the (product) construction from Sect.~\ref{sec:singleCooperationLevel} from a cooperation level specification with conjuncts in $D$ over the linear-time specifications $\mathcal{A}$, $\mathcal{G}$, $\mathcal{A} \rightarrow \mathcal{G}$, and $\mathcal{A} \wedge \mathcal{G}$. Let those linear-time specifications be represented by Rabin word automata with the state sets $Q_{\mathcal{A}}$, $Q_{\mathcal{G}}$, $Q_{\mathcal{A} \rightarrow \mathcal{G}}$, and $Q_{\mathcal{A} \wedge \mathcal{G}}$, respectively. Without loss of generality, let these state sets be disjoint.
We have that:
\begin{enumerate}
\item For all reachable states $q$ in $\mathcal{T}^H$, there is at most one state in $\mathit{unpack}(q)$ from each of $Q_{\mathcal{A}}$, $Q_{\mathcal{G}}$, $Q_{\mathcal{A} \rightarrow \mathcal{G}}$, and $Q_{\mathcal{A} \wedge \mathcal{G}}$.
\item All states in $\mathcal{T}^H$ with the same set $\mathit{unpack}(q)$ have the same languages.
\end{enumerate}
\end{lemma}
\begin{proof}
The first claim follows directly from the constructions from Sec.~\ref{sec:singleCooperationLevel}: for every state in the individual Rabin tree automata $\mathcal{T}^D$ built from cooperation level specification conjuncts of the shapes $\LTLG \langle E \rangle \varphi$ and $\varphi$, the automata always track the state of the corresponding word automata for a branch of the tree.

For the second claim, we decompose the states in $\mathcal{T}^H$ into their factors and prove the claim for each factor individually. For factors originating from cooperation level specifications of the shape $\varphi$ for some linear-time property $\varphi$, this fact is trivial. For factors originating from specification conjuncts of the shape $\LTLG \langle E \rangle \varphi$, the claim follows from the fact that the tree automaton states that only differ in their Boolean flag have the same successor state functions.
\end{proof}
Lemma~\ref{lem:productTreeAutomatonProperties} tells us how we can switch between cooperation levels. 
Assume that we can always read off \emph{all} current states of $Q_{\mathcal{A}}$, $Q_{\mathcal{G}}$, $Q_{\mathcal{A} \rightarrow \mathcal{G}}$, and $Q_{\mathcal{A} \wedge \mathcal{G}}$ from $\mathit{unpack}(q)$ for any state $q$ of a product tree automaton $\mathcal{T}^H$. 
This assumption can be made satisfied by letting $\mathcal{T}^H$ be the product of \emph{all} elements in the set of considered cooperation level specification conjuncts $D$, but only \changedOne{using} the ones in the current cooperation level $H$ when building the acceptance condition of $\mathcal{T}^H$.
Now consider a second cooperation level specification $H'$ that is higher in the hierarchy than $H$ and its associated tree automaton ${\mathcal{T}^H}'$. Let $W$ be the states in $\mathcal{T}^H$ with a non-empty language and $W'$ be the states of ${\mathcal{T}^H}'$ with a non-empty language. 
If we find a state $q'$ in ${\mathcal{T}^H}'$ for which $\mathit{unpack}(q') = \mathit{unpack}(q)$, and state $q'$ has a non-empty language, then we can simply re-route every transition to $q$ to $q'$ and obtain a new tree automaton with the states in $\mathcal{T}^H$ and ${\mathcal{T}^H}'$ that enforces a higher cooperation level on a bobble tree along all branches in run trees that lead to $q$. If we now identify the non-empty tree automaton states for all cooperation level specifications in our hierarchy, and apply this approach to all pairs of the corresponding tree automata \changedOne{and all of their states}, we end up with a tree automaton that accepts maximally cooperative computation trees. More formally, this line of reasoning shows the correctness of the following construction:
\begin{definition}
\label{def:buildingMaximallyCooperativeRabinTreeAutomaton}
Let $H_1, \ldots, H_{14}$ be the cooperation level specifications of our hierarchy, ordered by preference and respecting the hierarchy's partial order $\leq_H$, and let $\mathcal{T}^H_1, \ldots, \mathcal{T}^H_{14}$ be the non-deterministic Rabin tree automata for them. Let us furthermore rename states $q$ in an automaton $\mathcal{T}^H_j$ to $(q,j)$ to make their names unique, and let every tree automaton $\mathcal{T}^H_j$ be given as a tuple $(Q_j,\mathcal{I},\mathcal{O},\delta_j,q_{0,j},\mathcal{F}_j)$. Let $W \subseteq \bigcup_j Q_j$ be the states in the tree automata with a non-empty language.
We define the Rabin tree automaton $\mathcal{T} = (Q',\mathcal{I},\mathcal{O},\delta',q'_0,\mathcal{F}')$ encoding the maximally cooperative synthesis problem as follows:
\allowdisplaybreaks
\begin{align*}
Q' & = \bigcup_{j \in \{1, \ldots, 14\}} Q_j \\
q'_0 & = q_{0,j} \text { \ for \ } j = \max \{j \in \{1, \ldots, 14\} \mid q_{0,j} \in W \} \\
\mathcal{F}' & = \mathcal{F}_1 \cup \ldots \cup \mathcal{F}_{14} \\
\delta'((q,j),o) & = \{  \{ i \mapsto (q'',j') \mid j' = \max \{ k \in \{1, \ldots, 14\} \mid \exists (q'',k) \in Q_k : 
\\ & \quad \quad q'' \in W, \mathit{unpack}(q'') = \mathit{unpack}(f(i)), H_j \leq_H H_k \}, \\ & \quad \quad \mathit{unpack}(q'') = \mathit{unpack}(f(i)), \\ & \quad \quad ((j=j') \rightarrow q'' = f(i))\} \mid f \in \delta_j(q,o) \} \\
& \quad \quad \text{for all } (q,j) \in Q' \text{ and } i \in \mathcal{I}
\end{align*}
\end{definition}
\begin{theorem} \label{th:final}
A Rabin tree automaton built from linear-time specifications $\mathcal{A}$ and $\mathcal{G}$ according to Def.~\ref{def:buildingMaximallyCooperativeRabinTreeAutomaton} encodes the maximally cooperative synthesis problem for the specification $(\mathcal{A}, \mathcal{G})$. Building the tree automaton and checking it for emptiness can be performed in doubly-exponential time for specifications in LTL. 

\noindent
For specifications given as deterministic Rabin word automata, the time complexity is polynomial in the number of states and exponential in the number of acceptance pairs.
\end{theorem}
\begin{proof}
For the correctness, note that the tree automaton can switch between cooperation levels only finitely often, and whenever it switches, it only does so to strictly higher levels of cooperation.

To obtain doubly-exponential time complexity of maximally cooperative synthesis from LTL specifications, we first translate $\mathcal{A}$, $\mathcal{G}$, $\mathcal{A} \rightarrow \mathcal{G}$, and $\mathcal{A} \wedge \mathcal{G}$ to deterministic Rabin word automata with a doubly-exponential number of states and a singly-exponential number of acceptance condition pairs, which takes doubly-exponential time. The overall sizes of the tree automata build for the cooperation levels are then polynomial in the sizes of the Rabin word automata. We can compute $W$ in time exponential in the number of acceptance pairs and polynomial in the number of tree automaton states, which sums up to doubly-exponential time (in the lengths of $\mathcal{A}$ and $\mathcal{G}$) for LTL. When building $\mathcal{T}$ and computing the accepted computation tree, the same argument applies.
\changedOne{
If the specification is given in form of deterministic Rabin word automata, then all the product automata computed in the process have a number of states that is polynomial in the number of states of the input automata and a number of acceptance pairs that is polynomial in the number of acceptance pairs of the input automata. 
By the complexity of checking Rabin tree automata for emptiness and computing $W$, the second claim follows as well.

}
\end{proof}

\changedOne{Theorem~\ref{th:final} states that synthesizing maximally cooperative implementations from LTL specifications does not have a higher complexity than LTL synthesis in general. Note, however, that both the synthesis time and the size of the resulting systems can increase in practice as the Rabin automata built in cooperative synthesis are larger than in standard LTL synthesis.}

Also note that we can extend the theory from this subsection to also include cooperation level specification conjuncts of the shape $\langle E \rangle \varphi$. However, we would need to add flags to the tree automata to keep track of whether the current branch in a computation tree is the one on which $\varphi$ should hold. As these flags need to be tracked along changes between the cooperation levels, the definitions from this subsection would become substantially more complicated. Thus, we refrained from doing so here.

\section{Conclusion} \label{sec:concl}

Conventional synthesis algorithms often produce systems that 
actively work towards violating environment assumptions rather than satisfying assumptions and guarantees together.
In this paper, we worked out a fine-grained hierarchy of cooperation levels between the system and the environment for satisfying both guarantees \emph{and} assumptions as far as possible.
We also presented a synthesis procedure that maximizes the cooperation level in the hierarchy for linear-time specifications, such as Linear Temporal Logic (LTL).
The worst-case complexity of this procedure for LTL is the same as of conventional LTL synthesis.
Our approach relieves the user from requiring cooperation in the specification explicitly, which helps to keep the specification clean and abstract.

In the future, we plan to work out cooperative synthesis procedures for other 
specification languages, and evaluate the results on industrial applications. 

\bibliography{bib}

\begin{thebibliography}{10}

\bibitem{DBLP:conf/icalp/AlmagorBK13}
S.~Almagor, U.~Boker, and O.~Kupferman.
\newblock Formalizing and reasoning about quality.
\newblock In {\em ICALP}, pages 15--27. Springer, 2013.

\bibitem{DBLP:journals/jacm/AlurHK02}
R.~Alur, T.~A. Henzinger, and O.~Kupferman.
\newblock Alternating-time temporal logic.
\newblock {\em J. {ACM}}, 49(5):672--713, 2002.

\bibitem{DBLP:conf/stacs/Berwanger07}
D.~Berwanger.
\newblock Admissibility in infinite games.
\newblock In {\em {STACS}}, pages 188--199, 2007.

\bibitem{DBLP:journals/acta/BloemCGHHJKK14}
R.~Bloem, K.~Chatterjee, K.~Greimel, T.~A. Henzinger, G.~Hofferek,
  B.~Jobstmann, B.~K{\"{o}}nighofer, and R.~K{\"{o}}nighofer.
\newblock Synthesizing robust systems.
\newblock {\em Acta Inf.}, 51(3-4):193--220, 2014.

\bibitem{DBLP:conf/cav/BloemCHJ09}
R.~Bloem, K.~Chatterjee, T.~A. Henzinger, and B.~Jobstmann.
\newblock Better quality in synthesis through quantitative objectives.
\newblock In {\em {CAV}}, pages 140--156, 2009.

\bibitem{DBLP:journals/corr/BloemEJK14}
R.~Bloem, R.~Ehlers, S.~Jacobs, and R.~K{\"{o}}nighofer.
\newblock How to handle assumptions in synthesis.
\newblock In {\em {SYNT}}, pages 34--50, 2014.

\bibitem{atva15}
R.~Bloem, R.~Ehlers, and R.~K\"onighofer.
\newblock Cooperative reactive synthesis.
\newblock In {\em ATVA}. Springer, 2015.
\newblock To appear.

\bibitem{DBLP:journals/jcss/BloemJPPS12}
R.~Bloem, B.~Jobstmann, N.~Piterman, A.~Pnueli, and Y.~Sa'ar.
\newblock Synthesis of reactive(1) designs.
\newblock {\em J. Comput. Syst. Sci.}, 78(3):911--938, 2012.

\bibitem{Bozzelli}
Laura Bozzelli.
\newblock The complexity of {CTL}* + linear past.
\newblock In {\em Foundations of Software Science and Computational Structures
  ({FOSSACS})}, pages 186--200, 2008.

\bibitem{DBLP:conf/csl/BrenguierRS14}
R.~Brenguier, J.{-}F. Raskin, and M.~Sassolas.
\newblock The complexity of admissibility in omega-regular games.
\newblock In {\em {CSL-LICS}}, page~23, 2014.

\bibitem{DBLP:conf/vmcai/Chatterjee0FR14}
K.~Chatterjee, L.~Doyen, E.~Filiot, and J.{-}F. Raskin.
\newblock Doomsday equilibria for omega-regular games.
\newblock In {\em {VMCAI}}, pages 78--97, 2014.

\bibitem{DBLP:conf/tacas/ChatterjeeH07}
K.~Chatterjee and T.~A. Henzinger.
\newblock Assume-guarantee synthesis.
\newblock In {\em {TACAS}}, 2007.

\bibitem{DBLP:journals/iandc/ChatterjeeHP10}
K.~Chatterjee, T.~A. Henzinger, and N.~Piterman.
\newblock Strategy logic.
\newblock {\em Inf. Comput.}, 208(6):677--693, 2010.

\bibitem{DBLP:conf/rv/EhlersF11}
R.~Ehlers and B.~Finkbeiner.
\newblock Monitoring realizability.
\newblock In {\em {RV}}, pages 427--441, 2011.

\bibitem{iros15}
R.~Ehlers, R.~K\"onighofer, and R.~Bloem.
\newblock Synthesizing cooperative reactive mission plans.
\newblock In {\em IROS}. {IEEE}, 2015.

\bibitem{DBLP:conf/hybrid/EhlersT14}
R.~Ehlers and U.~Topcu.
\newblock Resilience to intermittent assumption violations in reactive
  synthesis.
\newblock In {\em {HSCC}}, pages 203--212, 2014.

\bibitem{DBLP:conf/mfcs/Faella09}
M.~Faella.
\newblock Admissible strategies in infinite games over graphs.
\newblock In {\em {MFCS}}, 2009.

\bibitem{DBLP:conf/tacas/FismanKL10}
D.~Fisman, O.~Kupferman, and Y.~Lustig.
\newblock Rational synthesis.
\newblock In {\em {TACAS}}, 2010.

\bibitem{DBLP:conf/popl/PnueliR89}
A.~Pnueli and R.~Rosner.
\newblock On the synthesis of a reactive module.
\newblock In {\em POPL}, 1989.

\end{thebibliography}

\extendedText{

\appendix
\clearpage


\section{Another Example}
\label{sec:app:ex}

To complement Section~\ref{sec:ex}, this section illustrates \emph{all} 
cooperation levels from Fig.~\ref{fig:hierarchy} on a more technical example.  
The assumptions $\Ass$ are defined using the Rabin word automaton 
$\mathcal{R}_\Ass = (Q,\allowbreak \inalph \times \outalph,\allowbreak 
\delta,\allowbreak q_0,\allowbreak \{(\emptyset,G_\Ass)\})$ shown in 
Fig.~\ref{fig:example}, where $Q=\{q_0,\ldots,q_{16}\}$, the input alphabet is 
$\inalph=\{x_0,x_1,x_2\}$, and the output alphabet is 
$\outalph=\{y_0,\ldots,y_{12}\}$.  Fig.~\ref{fig:example} labels edges with 
conditions over the input and output letters to define the transition function 
$\delta$. For instance, the condition $\neg x_0$ means that the transition is taken for 
all letters $(x,y)\in \inalph \times \outalph$ where $x \neq x_0$.  Edges 
without a label are always taken. 
Note 
that the system can only influence the next state from $q_2$: with output $y_i$, 
the next state will be $q_i$.  The acceptance condition is defined with $G_\Ass 
=\{q_0, q_6, q_7, q_8, q_9, q_{11}, q_{12}\}$. That is, the blue states in 
Fig.~\ref{fig:example} must be visited infinitely often for $\Ass$ to be 
satisfied.  The guarantees $\Gua$ are defined using the automaton 
$\mathcal{R}_\Gua = (Q,\allowbreak \inalph \times \outalph,\allowbreak 
\delta,\allowbreak q_0,\allowbreak \{(\emptyset,G_\Gua)\})$, which differs from 
$\mathcal{R}_\Ass$ only in the acceptance condition: with $G_\Gua = 
\{q_2,q_{13},q_{14},q_{15}\}$, $\Gua$ is satisfied if the green states in 
Fig.~\ref{fig:example} are visited infinitely often.

\begin{figure}
\centering
\begin{tikzpicture}[auto]
\node[state,draw,inner sep=1pt,minimum size=6mm]
  at (0,-0.8)   (S3) {$q_3$};
\node[state,initial left,draw,inner sep=1pt,minimum size=6mm,fill=blue!60!white]
  at (-1.5,-0.8)   (S0) {$q_0$};
\node[state,draw,inner sep=1pt,minimum size=6mm,fill=green!70!white]
  at (-1.5,-2)   (S2) {$q_2$};
\node[state,draw,inner sep=1pt,minimum size=6mm]
  at (0,-1.5)   (S4) {$q_4$};
\node[state,draw,inner sep=1pt,minimum size=6mm]
  at (-3,-1.5)  (S1) {$q_1$};
\node[state,draw,inner sep=1pt,minimum size=6mm]
  at (0,-3)   (S5) {$q_5$};
\node[state,draw,inner sep=1pt,minimum size=6mm,fill=blue!60!white]
  at (1.5,-3)   (S6) {$q_6$};
\node[state,draw,inner sep=1pt,minimum size=6mm,fill=blue!60!white]
  at (3,-3)   (S7) {$q_7$};
\node[state,draw,inner sep=1pt,minimum size=6mm,fill=green!70!white]
  at (2.25,-4) (S13) {$q_{13}$};
\node[state,draw,inner sep=1pt,minimum size=6mm,fill=blue!60!white]
  at (-1.5,-3)  (S9) {$q_9$};
\node[state,draw,inner sep=1pt,minimum size=6mm]
  at (-3,-3)  (S10) {$q_{10}$};
\node[state,draw,inner sep=1pt,minimum size=6mm,fill=blue!60!white]
  at (-4.5,-3)  (S8) {$q_8$};
\node[state,draw,inner sep=1pt,minimum size=6mm,fill=green!70!white]
  at (-4.5,-4) (S14) {$q_{14}$};
\node[state,draw,inner sep=1pt,minimum size=6mm,fill=blue!60!white]
  at (-6,-3)  (S11) {$q_{11}$};
\node[state,draw,inner sep=1pt,minimum size=6mm,fill=green!70!white]
  at (-6,-4) (S15) {$q_{15}$};
\node[state,draw,inner sep=1pt,minimum size=6mm,fill=blue!60!white]
  at (-7.5,-3)  (S12) {$q_{12}$};
\node[state,draw,inner sep=1pt,minimum size=6mm]
  at (-7.5,-4) (S16) {$q_{16}$};

\path
(S0) edge[->,bend angle=20, bend left] (S2)
(S1) edge[->] node[xshift=-0.5mm, yshift=-1.5mm] {$\neg x_0$} (S0)
(S1) edge[->,bend angle=20, bend left] node[xshift=-3mm, yshift=-0.8mm] {$x_0$} (S2)
(S2) edge[->,bend angle=20, bend left] node[xshift=1mm, yshift=1mm] {$y_0$} (S0)
(S2) edge[->,bend angle=20, bend left] node[xshift=3.5mm, yshift=3.5mm] {$y_1$} (S1)
(S2) edge[->,out=-35, in=-65, loop] node[xshift=-2mm, yshift=0.5mm] {$y_2$} (S2)
(S2) edge[->,bend angle=10, bend left] node[xshift=2.5mm, yshift=-0.5mm] {$y_3$} (S3)
(S2) edge[->,bend angle=20, bend right] node[xshift=2mm, yshift=-1mm] {$y_4$} (S4)
(S2) edge[->,bend angle=0, bend left] node[xshift=1mm, yshift=-2.5mm] {$y_5$} (S5)
(S2) edge[->,bend angle=11, bend left] node[xshift=7.8mm, yshift=-5mm] {$y_6$} (S6)
(S2) edge[->,bend angle=11, bend left] node[xshift=11.5mm, yshift=-5mm] {$y_7$} (S7)
(S2) edge[->,out=183, in=40] node[xshift=-12mm, yshift=-0.5mm] {$y_{8}$} (S8)
(S2) edge[->,bend angle=0, bend left] node[xshift=-4.5mm, yshift=0mm] {$y_9$} (S9)
(S2) edge[->,bend angle=0, bend left] node[xshift=-8mm, yshift=1.5mm] {$y_{10}$} (S10)
(S2) edge[->,out=183, in=40] node[xshift=-21mm, yshift=-1.5mm] {$y_{11}$} (S11)
(S2) edge[->,out=183, in=40] node[xshift=-22mm, yshift=-1.5mm] {$y_{12}$} (S12)
(S3) edge[->, loop right] node[] {$\neg x_0$} (S3)
(S3) edge[->] node[xshift=0mm, yshift=3.5mm] {$x_0$} (S0)
(S4) edge[->, loop right] node[] {$\neg x_0$} (S4)
(S4) edge[->,bend angle=-15, bend left] node[xshift=0.5mm, yshift=4mm] {$x_0$} (S2)
(S5) edge[->, loop below] (S5)
(S6) edge[->, loop right] node[xshift=-1mm, yshift=0mm] {$x_0$} (S6)
(S6) edge[->] node[xshift=-1mm, yshift=-1mm] {$x_1$} (S13)
(S6) edge[->] node[xshift=0mm, yshift=0.5mm] {$x_2$} (S5)
(S7) edge[->, loop right] node[xshift=-4mm, yshift=2.5mm] {$x_0$} (S7)
(S7) edge[->] node[xshift=0mm, yshift=2.5mm] {$\neg x_0$} (S13)
(S8) edge[->, loop right] node[xshift=-0.5mm, yshift=0mm] {$x_0$} (S8)
(S8) edge[->,bend angle=40, bend left] node[xshift=-0.5mm, yshift=0mm] {$\neg x_0$} (S14)
(S14) edge[->,bend angle=40, bend left] (S8)
(S9) edge[->, loop below] (S9)
(S10) edge[->, loop below] node[xshift=0mm, yshift=0.5mm] {$x_0$} (S10)
(S10) edge[->] node[xshift=0mm, yshift=-4mm] {$\neg x_0$} (S9)
(S11) edge[->, loop right] node[xshift=-2.5mm, yshift=1mm] {$\neg x_0$} (S11)
(S11) edge[->,bend angle=40, bend left] node[xshift=-0.5mm, yshift=0mm] {$x_0$} (S15)
(S15) edge[->,bend angle=40, bend left] node[xshift=1mm, yshift=0mm] {$x_0$} (S11)
(S15) edge[->, loop right] node[xshift=-2.5mm, yshift=1mm] {$\neg x_0$} (S15)
(S12) edge[->, loop right] node[xshift=-2.5mm, yshift=1mm] {$\neg x_0$} (S12)
(S12) edge[->,bend angle=40, bend left] node[xshift=-0.5mm, yshift=0mm] {$x_0$} (S16)
(S16) edge[->,bend angle=40, bend left] node[xshift=5mm, yshift=0mm] {$x_0$} (S12)
(S16) edge[->, loop right] node[xshift=-2.5mm, yshift=1mm] {$\neg x_0$} (S16)
(S13) edge[->, loop right] (S13)
;
\end{tikzpicture}
\caption{Example specification to illustrate all cooperation levels from
Fig.~\ref{fig:hierarchy}.}
\label{fig:example}
\end{figure}
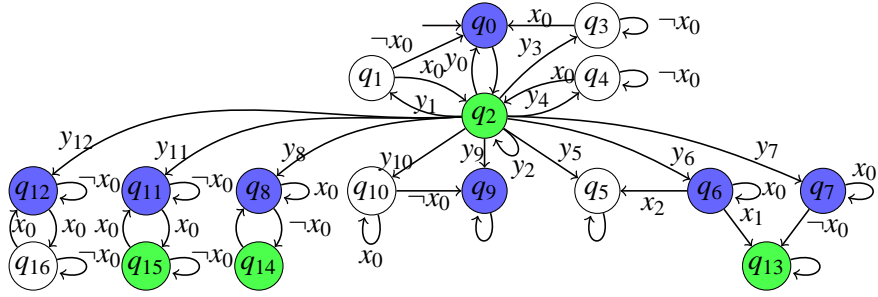

The following paragraphs present and discuss one system behavior (defined in the 
form of a computation tree) per cooperation level in Fig.~\ref{fig:hierarchy}.

\underline{\smash{$\Ass \rightarrow \Gua$:}}
The computation tree $\langle \inalph^*, \tau_5 \rangle$ with 
$\tau_5(w^\mathcal{I}) = y_5$ for all $w^\mathcal{I} \in \inalph^*$ represents a 
system that always outputs $y_5$.  It satisfies $\Ass \rightarrow \Gua$ because 
$q_5$ is a trap where neither $\Ass$ nor $\Gua$ is satisfied.  This is 
``correct'' but rather unsatisfactory.

\underline{\smash{$(\Ass \rightarrow \Gua) \wedge \reachinf \Gua$:}} 
is satisfied by the computation tree $\langle \inalph^*, \tau_4 \rangle$ with 
$\tau_4(w^\mathcal{I}) = y_4$.  Again, $\Ass \rightarrow \Gua$ holds because 
staying in $\{q_2,q_4\}$ violates $\Ass$.  Still, at any point in time, there 
exists some future environment behavior that satisfies $\Gua$ (namely one that 
choses $x_0$ infinitely often).  Thus, $\tau_4$ is slightly better than 
$\tau_5$.

\underline{\smash{$(\Ass \rightarrow \Gua) \wedge \reachinf \Ass$:}} 
is satisfied by $\langle \inalph^*, \tau_3 \rangle$ with $\tau_3(w^\mathcal{I}) 
= y_3$:  At any point in time, $\Ass$ can be satisfied for some future 
environment behavior (namely if $x_0$ is chosen in $q_3$ infinitely often).  
Since $\Ass \rightarrow \Gua$ (because of the unconditional edge from $q_0$ to 
$q_2$), this implies that also $\Gua$ can be reached.  Yet, neither $\Ass$ nor 
$\Gua$ are enforced because the environment could always give $x_1$ to get stuck 
in $q_3$.  However, assuming that the environment does not behave totally 
self-destructive, both $\Ass$ and $\Gua$ will be satisfied.  Thus, $\tau_3$ is 
even better than $\tau_4$.

\underline{\smash{$\Gua$:}} 
The computation tree with $\tau_2(w^\mathcal{I}) = y_2$ enforces $\Gua$, but 
defeats any hope of reaching $\Ass$. It is thus better than $\tau_2$ but 
incomparable with $\tau_3$: Enforcing $\Gua$ is better than having $\Gua$ 
reachable, but $\Ass$ being unreachable is worse than having $\Ass$ reachable. 
Hence, the choice between $\tau_2$ and $\tau_3$ is a question of selfishness.

\underline{\smash{$\Gua \wedge \reachinf \Ass$:}} 
is satisfied by $\tau_1(w^\mathcal{I}) = y_1$: No matter if the environment 
picks $x_0$ or not in $q_1$, the state $q_2$ will always be visited because of 
the unconditional edge from $q_0$ to $q_2$.  Visiting $q_0$ infinitely often 
(and thus satisfying $\Ass$) is always possible, but not enforced. Hence, 
$\tau_1$ dominates both $\tau_2$ and $\tau_3$.

\underline{\smash{$\Gua \wedge \Ass$:}} 
is enforced by $\tau_0(w^\mathcal{I}) = y_0$: independent of the environment, 
both $q_0$ and $q_2$ are visited repeatedly.  Thus, $\tau_0$ is the most 
desirable system behavior.

The remaining cooperation levels from the hierarchy in Fig.~\ref{fig:hierarchy} 
do not satisfy the traditional correctness criterion $\Ass \rightarrow \Gua$.  
Still, a computation tree for such a level behaves better than arbitrarily in 
various ways, and are useful from states where $\Ass \rightarrow \Gua$ cannot be
enforced.  

\underline{\smash{$\reachinf(\Ass \rightarrow \Gua)$:}}
is satisfied by the computation tree $\langle \inalph^*, \tau_6 \rangle$ with 
$\tau_6(w^\mathcal{I}) = y_6$ for all $w^\mathcal{I} \in \inalph^*$. At any 
point in time, there exists some environment behavior that either violates 
$\Ass$ (by going to $q_5$) or that satisfies $\Gua$ (by going to $q_{13}$). 
However, $\Gua$ alone is not reachable from every point in every execution: if 
the environment already traversed to $q_5$, $\Gua$ is lost.  Similarly, $\Ass$ 
is not reachable at every point in time because the environment may already have 
traversed to $q_{13}$.  Finally,  $\Ass \rightarrow \Gua$ is not enforced 
because with input $x_0$ the environment may stay in $q_6$ forever, thereby 
satisfying $\Ass$ but not $\Gua$.  Nevertheless, having $\Ass \rightarrow \Gua$ 
feasible at any point in time is better than nothing.

\underline{\smash{$\reachinf \Gua$:}}  
The computation tree with $\tau_7(w^\mathcal{I}) = y_7$ satisfies $\reachinf 
\Gua$, and is thus better than $\tau_6$. $\reachinf \Ass$ does not hold because 
from $q_{13}$, $\Ass$ is unreachable.  $\Ass \rightarrow \Gua$ is not enforced 
because the environment could stay in $q_7$ forever.  Yet, $\tau_7$ is still 
incomparable with $\tau_5$ (enforcing $\Ass \rightarrow \Gua$), because with 
$\tau_5$, there is not even a hope of reaching $\Gua$.

\underline{\smash{$\reachinf \Ass$:}} 
is similar to the previous case.  The computation tree with 
$\tau_{10}(w^\mathcal{I}) = y_{10}$ satisfies $\reachinf \Ass$, but $\Ass$ is 
not enforced because the environment can stay in $q_{10}$.  Moreover, 
$\reachinf(\Ass \rightarrow \Gua)$ does not hold because it is unreachable from 
$q_9$.

\underline{\smash{$\Ass$:}} 
The computation tree with $\tau_{9}(w^\mathcal{I}) = y_{9}$ enforces $\Ass$, and 
thus dominates $\tau_{10}$.  Yet all hopes for satisfying the guarantees are 
lost.

\underline{\smash{$\Ass \wedge \reachinf \Gua$:}}  
The system behavior $\tau_{8}(w^\mathcal{I}) = y_{8}$ enforces $\Ass \wedge 
\reachinf \Gua$ and  eliminates this defect of $\tau_9$: if the environment is 
not hostile and produces $x_1$ or $x_2$ infinitely often, $\Gua$ is achievable.

\underline{\smash{$(\reachinf \Ass) \wedge \reachinf(\Ass \rightarrow \Gua)$:}}
is satisfied by the system behavior $\tau_{12}(w^\mathcal{I}) = y_{12}$. The 
guarantee $\Gua$ is lost, but at any point in time, it is still possible that 
$\Ass$ is satisfied (by visiting $s_{12}$ infinitely often) and that $\Ass$ is 
violated (by only visiting $s_{16}$ from some point on).  This is slightly 
better than $\tau_{10}$ (which satisfies $\reachinf \Ass$ alone) because the 
correctness property $\Ass \rightarrow \Gua$ is reachable at any point.  With 
$\tau_{10}$, it may happen that the correctness property is lost once and for 
all.

\underline{\smash{$(\reachinf\Ass) \wedge (\reachinf\Gua)$}}
and
\underline{\smash{$\reachinf (\Ass \wedge \Gua)$}:}  
The system behavior $\tau_{11}(w^\mathcal{I}) = y_{11}$ satisfies 
$(\reachinf\Ass) \wedge (\reachinf\Gua)$ as well as $\reachinf (\Ass \wedge 
\Gua)$.  In our example, there is no way to distinguish these two levels because 
we chose objectives defined by visiting some states infinitely often.  For other 
specification classes, there can be a difference, though.  $\Ass$ is not 
enforced by $\tau_{11}$ because the environment could always stay in $q_{15}$. 
$\Ass \rightarrow \Gua$ is not enforced either because the environment could 
stick to $q_{11}$.  Still, this is better than a prospect of reaching only 
$\Ass$ or only $\Gua$.

\section{Extended Cooperation Hierarchies}\label{sec:app:hie}

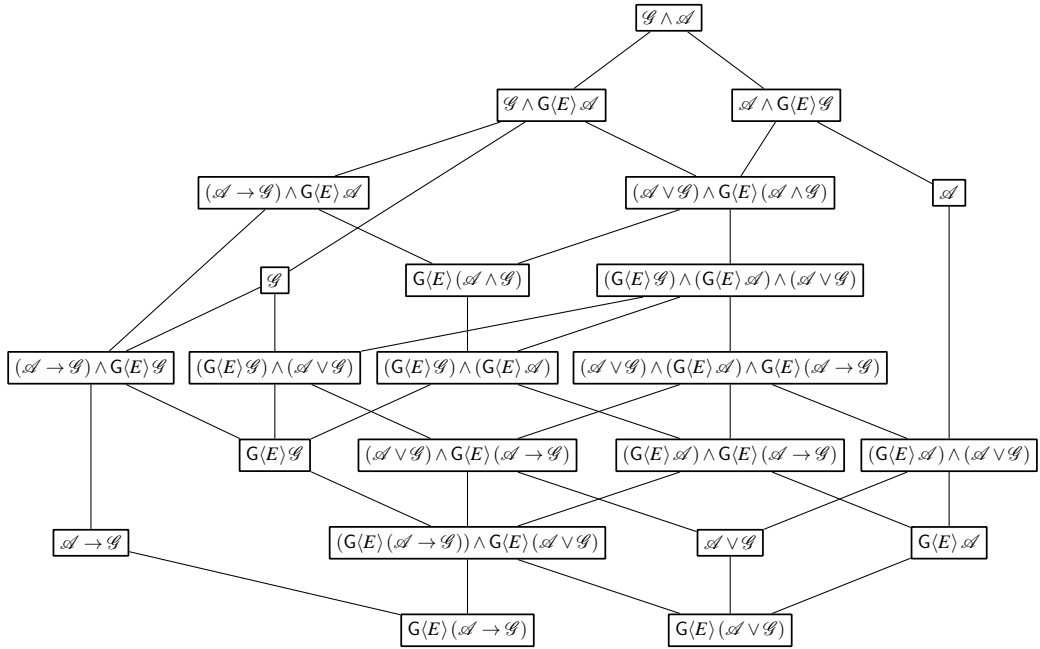
\begin{figure}[tb]
\centering
\scalebox{0.66}{
\begin{tikzpicture}[>=latex',line join=bevel,]
  \pgfsetlinewidth{1bp}
\pgfsetcolor{black}
\begin{scope}
  \definecolor{strokecol}{rgb}{0.0,0.0,0.0};
  \pgfsetstrokecolor{strokecol}
  \draw (405.0bp,50.0bp) node[rectangle,draw] (S0) {$\Ass \vee \Gua$};
\end{scope}
\begin{scope}
  \definecolor{strokecol}{rgb}{0.0,0.0,0.0};
  \pgfsetstrokecolor{strokecol}
  \draw (145.0bp,100.0bp) node[rectangle,draw] (S1) {$\reachinf \Gua$};
\end{scope}
\begin{scope}
  \definecolor{strokecol}{rgb}{0.0,0.0,0.0};
  \pgfsetstrokecolor{strokecol}
  \draw (530.0bp,50.0bp) node[rectangle,draw] (S2) {$\reachinf \Ass$};
\end{scope}
\begin{scope}
  \definecolor{strokecol}{rgb}{0.0,0.0,0.0};
  \pgfsetstrokecolor{strokecol}
  \draw (145.0bp,150.0bp) node[rectangle,draw] (S3) {$(\reachinf \Gua) \wedge (\Ass \vee \Gua)$};
\end{scope}
\begin{scope}
  \definecolor{strokecol}{rgb}{0.0,0.0,0.0};
  \pgfsetstrokecolor{strokecol}
  \draw (255.0bp,150.0bp) node[rectangle,draw] (S4) {$(\reachinf \Gua) \wedge (\reachinf \Ass)$};
\end{scope}
\begin{scope}
  \definecolor{strokecol}{rgb}{0.0,0.0,0.0};
  \pgfsetstrokecolor{strokecol}
  \draw (370.0bp,350.0bp) node[rectangle,draw] (S5) {$\Gua \wedge \Ass$};
\end{scope}
\begin{scope}
  \definecolor{strokecol}{rgb}{0.0,0.0,0.0};
  \pgfsetstrokecolor{strokecol}
  \draw (530.0bp,100.0bp) node[rectangle,draw] (S6) {$(\reachinf \Ass) \wedge (\Ass \vee \Gua)$};
\end{scope}
\begin{scope}
  \definecolor{strokecol}{rgb}{0.0,0.0,0.0};
  \pgfsetstrokecolor{strokecol}
  \draw (255.0bp,00.0bp) node[rectangle,draw] (S7) {$\reachinf (\Ass \rightarrow \Gua)$};
\end{scope}
\begin{scope}
  \definecolor{strokecol}{rgb}{0.0,0.0,0.0};
  \pgfsetstrokecolor{strokecol}
  \draw (150.0bp,250.0bp) node[rectangle,draw] (S8) {$(\Ass \rightarrow \Gua) \wedge \reachinf \Ass$};
\end{scope}
\begin{scope}
  \definecolor{strokecol}{rgb}{0.0,0.0,0.0};
  \pgfsetstrokecolor{strokecol}
  \draw (40.0bp,150.0bp) node[rectangle,draw] (S9) {$(\Ass \rightarrow \Gua) \wedge \reachinf \Gua$};
\end{scope}
\begin{scope}
  \definecolor{strokecol}{rgb}{0.0,0.0,0.0};
  \pgfsetstrokecolor{strokecol}
  \draw (530.0bp,250.0bp) node[rectangle,draw] (S10) {$\Ass$};
\end{scope}
\begin{scope}
  \definecolor{strokecol}{rgb}{0.0,0.0,0.0};
  \pgfsetstrokecolor{strokecol}
  \draw (303.0bp,300.0bp) node[rectangle,draw] (S11) {$\Gua \wedge \reachinf \Ass$};
\end{scope}
\begin{scope}
  \definecolor{strokecol}{rgb}{0.0,0.0,0.0};
  \pgfsetstrokecolor{strokecol}
  \draw (145.0bp,200.0bp) node[rectangle,draw] (S12) {$\Gua$};
\end{scope}
\begin{scope}
  \definecolor{strokecol}{rgb}{0.0,0.0,0.0};
  \pgfsetstrokecolor{strokecol}
  \draw (255.0bp,50.0bp) node[rectangle,draw] (S13) {$(\reachinf (\Ass \rightarrow \Gua)) \wedge \reachinf (\Ass \vee \Gua)$};
\end{scope}
\begin{scope}
  \definecolor{strokecol}{rgb}{0.0,0.0,0.0};
  \pgfsetstrokecolor{strokecol}
  \draw (255.0bp,200.0bp) node[rectangle,draw] (S14) {$\reachinf (\Ass \wedge \Gua)$};
\end{scope}
\begin{scope}
  \definecolor{strokecol}{rgb}{0.0,0.0,0.0};
  \pgfsetstrokecolor{strokecol}
  \draw (437.0bp,300.0bp) node[rectangle,draw] (S16) {$ \Ass \wedge \reachinf \Gua$};
\end{scope}
\begin{scope}
  \definecolor{strokecol}{rgb}{0.0,0.0,0.0};
  \pgfsetstrokecolor{strokecol}
  \draw (405.0bp,100.0bp) node[rectangle,draw] (S17)  {$(\reachinf \Ass) \wedge \reachinf (\Ass \rightarrow \Gua) $};
\end{scope}
\begin{scope}
  \definecolor{strokecol}{rgb}{0.0,0.0,0.0};
  \pgfsetstrokecolor{strokecol}
  \draw (40.0bp,50.0bp) node[rectangle,draw] (S18) {$\Ass \rightarrow \Gua$};
\end{scope}
\begin{scope}
  \definecolor{strokecol}{rgb}{0.0,0.0,0.0};
  \pgfsetstrokecolor{strokecol}
  \draw (405.0bp,200.0bp) node[rectangle,draw] (S19) {$(\reachinf \Gua) \wedge (\reachinf \Ass) \wedge (\Ass \vee \Gua)$};
\end{scope}
\begin{scope}
  \definecolor{strokecol}{rgb}{0.0,0.0,0.0};
  \pgfsetstrokecolor{strokecol}
  \draw (405.0bp,00.0bp) node[rectangle,draw] (S20) {$\reachinf (\Ass \vee \Gua)$};
\end{scope}
\begin{scope}
  \definecolor{strokecol}{rgb}{0.0,0.0,0.0};
  \pgfsetstrokecolor{strokecol}
  \draw (255.0bp,100.0bp) node[rectangle,draw] (S21) {$(\Ass \vee \Gua) \wedge \reachinf (\Ass \rightarrow \Gua)$};
\end{scope}
\begin{scope}
  \definecolor{strokecol}{rgb}{0.0,0.0,0.0};
  \pgfsetstrokecolor{strokecol}
  \draw (405.0bp,250.0bp) node[rectangle,draw] (S22) {$(\Ass \vee \Gua) \wedge \reachinf (\Ass \wedge \Gua)$};
\end{scope}
\begin{scope}
  \definecolor{strokecol}{rgb}{0.0,0.0,0.0};
  \pgfsetstrokecolor{strokecol}
  \draw (405.0bp,150.0bp) node[rectangle,draw] (S23) {$(\Ass \vee \Gua) \wedge (\reachinf \Ass) \wedge \reachinf (\Ass \rightarrow \Gua)$};
\end{scope}

\path
(S5) edge (S11)
(S5) edge (S16)
(S11) edge (S8)
(S11) edge (S12)
(S11) edge (S22)
(S16) edge (S22)
(S16) edge (S10)
(S22) edge (S14)
(S22) edge (S19)
(S8) edge (S9)
(S8) edge (S14)
(S10) edge (S6)
(S12) edge (S9)
(S12) edge (S3)
(S14) edge (S4)
(S19) edge (S3)
(S19) edge (S4)
(S19) edge (S23)
(S9) edge (S1)
(S3) edge (S1)
(S3) edge (S21)
(S4) edge (S1)
(S4) edge (S17)
(S23) edge (S21)
(S23) edge (S17)
(S23) edge (S6)
(S9) edge (S18)
(S1) edge (S13)
(S21) edge (S13)
(S21) edge (S0)
(S17) edge (S13)
(S17) edge (S2)
(S6) edge (S0)
(S6) edge (S2)
(S18) edge (S7)
(S13) edge (S7)
(S13) edge (S20)
(S0) edge (S20)
(S2) edge (S20)
;

\end{tikzpicture}
}
\caption{A refined hierarchy of cooperation levels.}
\label{fig:hie_or}
\end{figure}

As discussed in Section~\ref{sec:alt}, we can extend our hierarchy in various
ways.  When extending the conjuncts from 
Eqn.~\ref{eqn:cooperationLevelConjunctsConsidered} with $\Ass \vee \Gua$ and 
$\reachinf(\Ass \vee \Gua)$, we obtain the hierarchy shown in 
Figure~\ref{fig:hie_or}, which has $23$ levels.  We applied the following 
reduction rules in addition to those mentioned in Section~\ref{sec:defcop}:
\begin{itemize}
\item $\Gua$ implies $\Ass \vee \Gua$,
\item $\Ass$ implies $\Ass \vee \Gua$,
\item $\reachinf \Gua$ implies $\reachinf (\Ass \vee \Gua)$,
\item $\reachinf \Ass$ implies $\reachinf (\Ass \vee \Gua)$, 
\item $\Ass \rightarrow \Gua$ and $\Ass \vee \Gua$ together imply $\Gua$, and
\item $\Ass \rightarrow \Gua$ and $\reachinf (\Ass \vee \Gua)$ together imply $\reachinf \Gua$.
\end{itemize}

When further extending the conjuncts from 
Eqn.~\ref{eqn:cooperationLevelConjunctsConsidered} with $\reach(\Gua)$, 
$\reach(\Ass)$, $\reach(\Ass \wedge \Gua)$, $\reach(\Ass \rightarrow \Gua)$ and $\reach(\Ass \vee \Gua)$ we 
obtain a hierarchy of $77$ cooperation levels using the following additional
reduction rules:
\begin{itemize}
\item $\reachinf \varphi$ implies $\reach \varphi$ for all $\varphi \in \{\Gua, \Ass, \Ass \wedge \Gua, \Ass \rightarrow \Gua, \Ass \vee \Gua\}$
\item $\reach\mathcal{G}$ implies $\reach{(\mathcal{A} \rightarrow \mathcal{G})}$,
\item $\reach(\mathcal{A} \wedge \mathcal{G})$ implies $\reach\mathcal{A}$ and $\reachinf\mathcal{G}$,
\item $\mathcal{A} \rightarrow \mathcal{G}$ and $\reach \mathcal{A}$ together imply $\reach(\mathcal{A} \wedge \mathcal{G})$,
\item $\mathcal{A}$ and $\reach \mathcal{G}$ together imply $\reach(\mathcal{A} \wedge \mathcal{G})$,
\item $\reach(\mathcal{A} \rightarrow \mathcal{G})$ and $\mathcal{A}$ together imply $\reach(\mathcal{G})$.
\item $\reach \Gua$ implies $\reach (\Ass \vee \Gua)$,
\item $\reach \Ass$ implies $\reach (\Ass \vee \Gua)$,
\item $\Ass \rightarrow \Gua$ and $\reach (\Ass \vee \Gua)$ together imply $\reach \Gua$.
\end{itemize}
The first six reduction rules directly correspond to the rules from 
Section~\ref{sec:defcop} but using $\reach$ instead of $\reachinf$.  The latter 
three rules correspond to rules from the previous paragraph, again instantiated 
with $\reachinf$ instead of $\reach$.

}{} 

\end{document}